\documentclass[preprint,12pt,3p]{elsarticle}

\usepackage{amssymb}

\usepackage{amsthm}
\usepackage{graphicx}
\usepackage{amsmath}
\usepackage{hyperref}
\usepackage{graphicx}

\usepackage{pgf,tikz}

\usetikzlibrary{snakes}
\tikzstyle{printersafe}=[snake=snake,segment amplitude=0 pt]

\usetikzlibrary{arrows}
\usetikzlibrary{shapes}
\usepackage{multirow}
\usepackage{latexsym}
\usepackage{setspace}
\PassOptionsToPackage{boxed,section}{algorithm}
\usepackage{algorithm}
\usepackage{algorithmic}
\usepackage{enumerate}
\usepackage{amsmath}			
\usepackage{amssymb}			
\usepackage{url}
\usepackage{setspace}
\usepackage{tikz}
\usetikzlibrary{arrows}
\usetikzlibrary{shapes}
\usetikzlibrary{decorations.markings}
\usepackage{geometry}
\usepackage{bbm}

\newtheorem{theorem}{\em Theorem}

\newtheorem{lemma}{\em Lemma}

\newtheorem{corollary}{\em Corollary}
\newtheorem{observation}{\em Observation}

\journal{Journal}

\begin{document}

\begin{frontmatter}

\title{Reducing Maximum Weighted Matching to the Largest Cardinality Matching in {\bf CONGEST} \tnoteref{label0}}
\tnotetext[label0]{The work on this paper has been supported by the National Science Foundation through Grant CCF-2008422}

\author[label1]{Vahan Mkrtchyan}
\address[label1]{Computer Science Department, Boston College, Chestnut Hill, MA, USA-02467}

%
\ead{vahan.mkrtchyan@bc.edu}
%
%

\begin{abstract}
In this paper, we reduce the maximum weighted matching problem to the largest cardinality matching in {\bf CONGEST}. The paper presents two technical contributions. The first of them is a simple $poly(\log n, \frac{1}{\varepsilon}, t, \ln w_t)$-round {\bf CONGEST} algorithm for reducing the maximum weighted matching problem to the largest cardinality matching problem. This is achieved under the assumption that all vertices know all edge-weights $\{w_1,....,w_t\}$ (in particular, they know $t$, the number of different edge-weights), though a particular vertex may not know the weight of a particular edge. Our second ingredient is a simple rounding algorithm (similar to approximation algorithms for the bin packing problem) allowing to reduce general instances of the maximum weighted matching problem to ones satisfying the assumptions of the first ingredient, in which $t\leq poly'(\log n, \frac{1}{\varepsilon})$. We end the paper with a brief discussion of implementing our algorithms in {\bf CONGEST}. Our main conclusion is that we just need constant rounds for the reduction.
\end{abstract}


\end{frontmatter}



\section{Introduction}
\label{IntroSection}

Let $Z^+$ be the set of positive integers, and let $Q^+$ be the set of positive rational numbers. We will frequently use the notation $poly()$ to denote a polynomial. We will use it when we will not care about the coefficients and the degree of the polynomial. In the cases, when we will work with more than one such undefined polynomials, we will use notations $poly'()$, $poly''()$, etc. in order to explicitly state that they might be different.

In this paper, we consider finite, undirected graphs that do not contain loops or parallel edges. If $G$ is a graph, then let $V=V(G)$ and $E=E(G)$ be the sets of its vertices and edges, respectively. A matching in a graph is a subset of its edges such that every vertex is incident to at most one edge from the subset.

The focus of the present paper is on the largest cardinality matching problem, in which for a given graph $G$ one needs to find a largest matching. Let $\nu(G)$ be the size of a largest matching of $G$. Then, in this problem the goal is to come up with an algorithm whose output is a matching of size $\nu(G)$. In this paper, our focus is on obtaining approximation algorithms for this problem and its weighted extension stated below. Hence, in the input we assume that we are given a graph $G$, and a number $\varepsilon \in (0,1)$. The goal is to return a matching $M$ of $G$, such that $M$ contains at least $(1-\varepsilon)\cdot \nu(G)$ edges. Throughout the paper, we will assume that $\varepsilon$ is a rational number.

In the paper, we will consider the maximum weighted matching problem (denoted by MWM), in which in parallel to the given graph $G$, we will be given a number $w(e)\in Z^+$ called the weight of the edge $e\in E(G)$. The goal is to find a matching $M$ maximizing the weight of $M$, that is
\[w(M)=\sum_{e\in M}w(e).\]
The weight of the optimal matching will be denoted by $OPT(G)$. Clearly, the maximum weighted matching problem includes the largest cardinality matching problem, as if $w(e)=1$, we get the latter problem. Also, note that the maximum weighted matching problem can be considered as the non-uniform case of the largest cardinality matching problem, since if all edges of $G$ have the same weight $W$, then the weights are irrelevant and in particular we have
\[OPT(G)=W\cdot \nu(G).\]
In the paper, we will be interested in approximating this problem, hence in the input except the graph $G$, its edge-weights $w(e)$, we will be given a number $\varepsilon \in (0,1)$. The goal is to return a matching $M$ of $G$, such that the edges of $M$ together have weight at least $(1-\varepsilon)\cdot OPT(G)$. 

The following extension of MWM to rational weights will be important for us. We will denote it by Rational MWM, and its input, output and the goals will be the same as in MWM, except that the weights will be assumed to be rational numbers from the interval $[1, +\infty)$.

The first polynomial time algorithm for the largest cardinality matching problem and MWM, its weighted extension was proposed in \cite{E1965,E1965II}. More on these problems and their history can be found in \cite{DP14}, where, in particular, tables III, IV and references therein provide a summary of the state of art before 2014 both for the largest cardinality matching problem and MWM.

Approximation algorithms for polynomial-time solvable problems have been presented before. \cite{DH03} presents a linear time $\frac{1}{2}$-approximation algorithm for MWM. In \cite{DP14}, for any $\varepsilon>0$, a $(1-\varepsilon)$-approximation algorithm for MWM is presented that runs in time $O(m\cdot \frac{1}{\varepsilon}\cdot \log \frac{1}{\varepsilon})$. Here $m=|E(G)|$. See \cite{DHZ20} for an approximation algorithm for a spanning tree problem, in which one seeks to minimize the maximum degree of the tree. This problem is known to be polynomial time solvable thanks to \cite{Furer94}. A recent preprint \cite{KTY22} provides a $(1+\varepsilon)$-approximation algorithm for computing the degeneracy of the graph.

In this paper, we will be interested in designing algorithms for Rational MWM in the {\bf CONGEST} model from the area of distributed computing. In the {\bf CONGEST} model we view each vertex of the input graph as a processor, and in each round vertices (or processors) are allowed to exchange messages of length $O(\log n)$. Here $n$ is the number of vertices of the input graph. We will be interested in designing algorithms in this model that have small number of such rounds.

Let us note that in the {\bf CONGEST} model, generally it is required that all integer weights are bounded by $poly(n)$. Since, in the paper we will work with Rational MWM, where the weights might be rational numbers, we will view each rational number $r$ as a pair of integers $(p, q)$, where both $p$ and $q$ are bounded by $poly(n)$. We assume that all vertices know $\varepsilon$, $W$-the largest edge-weight and $poly(n)$-the upper bound mentioned above.

\cite{AK20} presents a distributed algorithm for verifying a given solution of the largest cardinality matching problem. See \cite{AKO18} for more on this. \cite{Kitamura2021} presents a randomized {\bf CONGEST} algorithm for finding a matching of size $\nu(G)$ whose number of rounds depends polynomially on $\log n$ and $\nu(G)$. \cite{FFK21} presents $poly(\log n, \frac{1}{\varepsilon})$-round $(1-\varepsilon)$-approximation algorithm for the Weighted Vertex Cover problem in bipartite graphs in {\bf CONGEST}. There, similar result is obtained for MWM with exponential dependence on $\frac{1}{\varepsilon}$ and no restriction on the input graph. This is improved in 
\cite{FMU21}, where the authors present a $poly(\log n, \frac{1}{\varepsilon})$-round $(1-\varepsilon)$-approximation algorithm for the Largest Cardinality Matching problem in {\bf CONGEST}. See \cite{LPSP15} for more on this. \cite{HYN22} presents a $poly(\log \log n)$-round randomized {\bf CONGEST} algorithm for $(1+\varepsilon)\cdot \Delta$-edge-coloring all graphs with $n$ vertices for sufficiently large constant maximum degree $\Delta$. See the recent preprint \cite{FHM22}, for results related to the classical Brooks' theorem on $\Delta$-coloring all graphs with maximum degree $\Delta$. Distributed algorithms for spanning tree and other problems can be found in \cite{ DHN18, GP19, Kutten98}. Finally, let us note that one can find new open problems in \cite{C21} for further research.

In this paper, we show a {\bf CONGEST} algorithm that allows to reduce Rational MWM to the largest cardinality matching problem. Combined with \cite{FMU21}, our result implies that there is a $poly(\log n, \frac{1}{\varepsilon})$-round $(1-\varepsilon)$-approximation algorithm for Rational MWM in {\bf CONGEST}. Let us note that the same result for MWM has been obtained by Su and Huang, independently in \cite{SH22}.

\section{Some auxiliary results}
\label{AuxResultsSection}

In this section, we present some auxiliary results that will be helpful for obtaining the main result of the paper. We start with a lemma that will be frequently used in various reductions.

\begin{lemma}
\label{LemmaEpsilonFracTreIncrease} Let $G$ be a graph equipped with a weight function $w:E\rightarrow Q^+$, and let $\varepsilon \in (0,1)$ be a rational number. Consider a new instance of MWM where the graph is the same, the weight function is $w'$, which coincides with $w$, except that all edges $e$ with smallest $w$-weight $w_1$, now have weight $w'_1$, where \[w_1\leq w'(e)=w'_1\leq \left(1+\frac{\varepsilon}{t}\right)\cdot w_1.\] Then if a matching $M$ of $G$ is such that
\[w'(M)\geq (1-\theta)\cdot OPT',\]
where $\theta=\frac{t-1}{t}\cdot \varepsilon$, then \[w(M)\geq (1-\varepsilon)\cdot OPT.\]
\end{lemma}

\begin{proof} Let $z=w'_1-w_1\leq \frac{\varepsilon}{t}\cdot w_1$. Assume that $M$ is a matching with 
\[w'(M)\geq \left(1-\theta\right)\cdot OPT',\]
where $OPT'$ is the optimum for $w'$, and $\theta=\frac{t-1}{t}\cdot \varepsilon$. We have
\begin{align*}
    w(M) &\geq w'(M)-z\cdot \nu(G)\geq \left(1-\theta\right)\cdot OPT'-\frac{z}{w_1}\cdot w_1\cdot \nu(G)\\
         &\geq \left(1-\theta\right)\cdot OPT-\frac{z}{w_1}\cdot OPT=\left(1-\theta-\frac{z}{w_1}\right)\cdot OPT\geq \left(1-\theta-\frac{\varepsilon}{t}\right)\cdot OPT\\
         &=(1-\varepsilon)\cdot OPT,
\end{align*} since
\[w_1\cdot \nu(G)\leq OPT\leq OPT'\]
and
\[\theta+\frac{\varepsilon}{t}=\frac{t-1}{t}\cdot \varepsilon+\frac{\varepsilon}{t}=\varepsilon.\]
The proof is complete.
\end{proof}

\begin{corollary}
\label{corollaryEpsilonFracT} Suppose that $G$ is a graph, $w:E\rightarrow Q^+$ is a weight function taking values $w_1<w_2<...<w_t$. Assume that 
\[w_2-w_1\leq \frac{\varepsilon}{t}\cdot w_1.\]
Then, if we consider a new instance of MWM where the graph is the same, the edge weight function is $w'$, which coincides with $w$ except that edges of $w$-weight $w_1$ now have $w'$-weight $w_2$. Then if $M$ is a matching of $G$ with
\[w'(M)\geq (1-\theta)\cdot OPT',\]
where $\theta=\frac{t-1}{t}\cdot \varepsilon$, then
\[w(M)\geq (1-\varepsilon)\cdot OPT.\]
\end{corollary}

We will use the following lemma frequently.
\begin{lemma}
\label{lemmaRoundingUp} Let $(G, w, \varepsilon\in (0,1))$ be an instance of MWM. For $\tau >1$ consider an instance $(G, w', \theta)$, where each edge $e$ of weight $w(e)$ with
\[ \tau^{i-1}< w(e) \leq \tau^{i},\]
is rounded to $\tau^{i}$. Then if $M$ is a matching with
\[w'(M)\geq (1-\theta)\cdot OPT',\]
then
\[w(M)\geq (1-\varepsilon)\cdot OPT,\]
provided that
 \[\frac{1}{\tau}\cdot (1-\theta)\geq 1-\varepsilon.\]
\end{lemma}

\begin{proof} For $j \geq 0$ let $B_j$ be the set of edges of $G$ with \[\tau^{j-1}<w(e)\leq \tau^{j}\] (and hence $w'(e)=\tau^j$). Let $\alpha_j=\tau^j$, and let $\beta_j$ be the smallest edge-weight in $B_j$. If we assume that $B_1,...,B_r$ are the only non-empty ones, we will have:
\begin{align*}
    w(M) &\geq \beta_1\cdot |M\cap B_1|+...+\beta_r\cdot |M\cap B_r|=\frac{\beta_1}{\alpha_1}\cdot \alpha_1 \cdot |M\cap B_1|+...+\frac{\beta_r}{\alpha_r}\cdot \alpha_r \cdot |M\cap B_r|\\
    &\geq \frac{1}{\tau}\cdot \left( \alpha_1 \cdot |M\cap B_1|+...+ \alpha_r \cdot |M\cap B_r| \right)=\frac{1}{\tau}\cdot w'(M)\geq \frac{1}{\tau}\cdot \left(1-\theta\right)\cdot OPT'\\
    &\geq (1-\varepsilon)\cdot OPT.
\end{align*} Here we used the simple inequality
\[OPT\leq OPT',\]
and
\[\alpha_j\leq \beta_j \cdot \tau\]
for $j=1,...,r$. In each group the ratio between largest and smallest edge-weight is at most $\tau$ by definition. The proof is complete.
\end{proof}

The following observation will be helpful.

\begin{observation}
\label{observationLogFunction} Consider the function \[ f(x)=\frac{1}{\ln \left(\frac{1}{1-x}\right)}=-\frac{1}{\ln (1-x)}\] 
on the interval $(0,1)$. Then, there is a positive $A>0$, such that
\[f(x)\leq  A+\frac{1}{x}\]
for every $x\in (0,1)$.
\end{observation}

\begin{proof} We provide a proof for the sake of completeness. Recall that
\[ \lim_{x\rightarrow 0}\frac{\ln (1+x)}{x}=1.\]
Hence,
\[\lim_{x\rightarrow 0}\frac{f(x)}{\frac{1}{x}}= \lim_{x\rightarrow 0} \frac{ \frac{1}{\ln (\frac{1}{1-x})} }{\frac{1}{x}}=\lim_{x\rightarrow 0} \frac{x}{-\ln ({1-x})}=\lim_{x\rightarrow 0} \frac{-x}{\ln ({1-x})}=1. \]
Thus, there is a positive $\delta > 0$, such that for any $x\in (0,\delta)$, we have
\[f(x)<\frac{1}{x}+1.\]
On the other hand, the function $f(x)$ is continuous on the interval $[\delta, 1)$ and
\[\lim_{x\rightarrow 1} f(x)=0. \]
Thus, $f(x)$ is bounded by some $A>0$ on $[\delta, 1)$. Hence, we have the statement. The proof is complete.
\end{proof}

The next lemma will play an important role in obtaining the main result of the paper. Throughout the paper we assume that {\bf LargestCardinalityMatching} is an algorithm for approximating the largest cardinality matching problem in {\bf CONGEST} in $poly(\log n, \frac{1}{\varepsilon})$ rounds. When we need to specify its input, we will write {\bf LargestCardinalityMatching}$(G, \varepsilon)$.

\begin{lemma}
\label{lemmatdifferentFIXEDEdgeWeightsCongest}  For all rational $\varepsilon\in (0,1)$ there exists a $poly(\log n, \frac{1}{\varepsilon}, t, \ln w_t)$-round $(1-\varepsilon)$-approximation algorithm for Rational-MWM in {\bf CONGEST} for all instances $(G, w, \varepsilon>0)$ with at most $t$ (rational) edge-weights \[1\leq w_1\leq w_2\leq ...\leq w_t,\] in the case when all vertices of $G$ know the numbers $\{w_1, w_2, ...,w_t\}$ though they may not know which edge has which weight. In particular, all vertices know $t$. 
\end{lemma} 

\begin{proof} Consider the Algorithm \ref{algAUX}.

 \begin{algorithm}[htbp]
  
  \caption{The algorithm {\bf MaximumWeightedMatching}$(G, w, \varepsilon, t )$.}\label{algAUX}
\begin{algorithmic}[1]
\STATE {\bf Input:} A graph $G$, a weight-function $ w:E(G)\rightarrow Q$, $ \varepsilon\in (0,1)\cap Q^+ $ and $t\in Z^+$.
\STATE {\bf Output:} A matching $M$ of the graph $G$.
\STATE {\bf Assumption:} We assume that the rational numbers $w_1, w_2, ...,w_t$ are known to all vertices of $G$. In particular, all vertices know $t$. Let us note that the vertices are not expected to know which edge has which weight.

\medskip

\medskip

\medskip

\IF{($t=1$)}
\STATE return {\bf LargestCardinalityMatching}$(G, \varepsilon)$.
\ENDIF
\STATE Set: $z=w_2-w_1$. 
\IF{$(\frac{z}{w_1}\leq \frac{\varepsilon}{t})$} 
\STATE Define a new edge weight function $w'$, which coincides with $w$ except that all edges of weight $w_{1}$ now have weight $w_2$. Set: $\varepsilon:=\frac{t-1}{t}\cdot \varepsilon$.
\STATE Return {\bf MaximumWeightedMatching}$(G, w', \varepsilon, t-1)$.
\ELSE

\STATE Define $x=\frac{1}{1-\frac{\varepsilon}{t}}$ and $\varepsilon:=\frac{t-1}{t}\cdot \varepsilon$. 

\STATE Each vertex $v$ of $G$ rounds the weight $w(e)$ of an edge $e$ incident to it to $w'(e)=x^j$, where $x^{j-1} < w(e) \leq x^j$. Let $t'$ be the upper bound for the number of different edge-weights. Set: $z:=w'_2-w'_1$.
\ENDIF

\IF{($z=0$)} 
\STATE return {\bf LargestCardinalityMatching}$(G, \varepsilon)$. 
\ENDIF

\IF{$(\frac{z}{w'_1} > \frac{\varepsilon}{t'})$}
    \STATE Define a new edge weight function $w''$, which coincides with $w'$ except that all edges of weight $w'_{1}$ now have weight $(1+\frac{\varepsilon}{t'})\cdot w'_1$. 
  \ELSE
    \STATE Define a new edge weight function $w''$, which coincides with $w'$ except that all edges of weight $w'_{1}$ now have weight $w'_2$. Take $t':=t'-1$.
  \ENDIF

%
\STATE Set: $z:=w''_2-w''_1$, $\varepsilon:=\frac{t'-1}{t'}\cdot \varepsilon$.
\STATE Go to STEP 15.

\end{algorithmic}
\end{algorithm}

Let us show that for any instance $(G, \varepsilon, w(e)\in \{w_1, w_2,...,w_t\})$ it always returns a matching $M$ with $w(M)\geq (1-\varepsilon)\cdot OPT$. Let us prove this statement by induction on $t$. Clearly, this is true for $t=1$. Assume that it is true for all instances with at most $t-1$ different edge-weights, and let us consider an instance with $t$ different edge-weights. Note that in STEP 19 we increase the smallest edge-weight of the current instance of MWM. Hence at some point, we will reach to an instance where the condition in STEP 18 will not be satisfied, which will decrease $t$-the number of different edge-weights in STEP 21, and hence inductive hypothesis can be applied. Thus, our algorithm will return a matching $M$ with \[ w(M)\geq (1-\varepsilon)\cdot OPT\] because of the inductive hypothesis, Lemma \ref{LemmaEpsilonFracTreIncrease}, Corollary \ref{corollaryEpsilonFracT} and Lemma \ref{lemmaRoundingUp}.

\medskip

Now let us find an upper bound for the number of rounds of our algorithm. Assume that $B$ and $C$ are constants such that the number of rounds for {\bf LargestCardinalityMatching} is at most
\[\leq B\left( \frac{\log n}{\varepsilon}\right)^C.\]
If the condition in STEP 4 is satisfied, then the number of rounds is upper bounded by
\[\leq 1+ poly\left(\log n, \frac{1}{\varepsilon}, t', \ln w_t\right)= 1+ poly\left(\log n, \frac{1}{\varepsilon}, t-1, \ln w_t\right)\]
since $t'=t-1$. Let us assume that the condition in STEP 4 is not satisfied. Then, after STEP 12 and STEP 13, we will have a new instance with weights $\{x, x^2,...,x^{t'}\}$, where \[x=\frac{1}{1-\frac{\varepsilon}{t}}=\frac{1}{1-\frac{\varepsilon'}{t'}}\leq 1+2\cdot \frac{\varepsilon}{t}=1+2\cdot \frac{\varepsilon'}{t-1},\] since $t\geq 2$,

\[x^{t'-1}<w_t\leq x^{t'},\] and \[\varepsilon'=\frac{t-1}{t}\cdot \varepsilon. \]
Then
\[t'< 1+ \log_{x} w_t = 1+ \frac{\ln w_t}{\ln x}\leq 1+\left(A+\frac{t}{\varepsilon}\right)\cdot \ln w_t\]
by Observation \ref{observationLogFunction}.

Now, note that because of the inequality in the definition of $x$, after { at most two rounds} in STEP 15, STEP 18 and STEP 19, the algorithm will reach STEP 21, where it will decrease the number of different edge-weights by one. Moreover, since the ratio between two consecutive edge-weights is $x$, after {at most two rounds} in STEP 15, STEP 18 and STEP 19, the algorithm will reach STEP 21, where it will decrease the number of edge-weights, again.
Thus, for the current $\varepsilon$ we have that in these rounds, it makes jumps according to the following scheme:
\[\varepsilon\rightarrow \frac{\varepsilon}{\frac{t'}{t'-1}}\rightarrow \frac{\varepsilon}{\left(\frac{t'}{t'-1}\right)^2}=\varepsilon\cdot \left(\frac{t'-1}{t'}\right)^{2}.\]
Hence
\[\varepsilon'\geq \left(\frac{t'-1}{t'}\right)^{2} \cdot \varepsilon.\]
Therefore, after at most $2t'$ rounds, we will reach STEP 16, where the largest cardinality matching algorithm will be applied with
\[ \varepsilon_{unweighted}\geq \left(\frac{1}{2}\right)^{2} \cdot \left(\frac{2}{3}\right)^{2}\cdot ...\cdot \left(\frac{t'-2}{t'-1}\right)^{2} \cdot  \left(\frac{t'-1}{t'}\right)^{2} \cdot \varepsilon=\frac{\varepsilon}{t'^2}.\]

Thus, if we have $t$ different edge-weights, we can safely say that there will be at most
\[\leq 2\cdot t+B\cdot \left(\frac{\log n}{\varepsilon_{unweighted}} \right)^C\]
rounds. Here, $\varepsilon_{unweighted}$ is the last value of $\varepsilon$, where the algorithm calls {\bf LargestCardinalityMatching}. We have
\[ \varepsilon_{unweighted}\geq \left(\frac{1}{t'}\right)^{2} \cdot \varepsilon. \]
Hence
\[\frac{1}{\varepsilon_{unweighted}}\leq \frac{t'^2}{\varepsilon}.  \]
Thus, the number of rounds will be upper bounded by
\[\leq 2t+B\cdot \left(\frac{\log n}{\varepsilon_{unweighted}} \right)^C \leq 2t+B\cdot \left(\frac{\log n}{\varepsilon}\cdot t'^2 \right)^C=poly\left(\log n, \frac{1}{\varepsilon}, t, \ln w_t\right).\]
The proof is complete. 
\end{proof}

\section{The main result}
\label{MainResultSection}

In this section, we obtain the main result of the paper.

\begin{theorem}
\label{theoremMain} For all $\varepsilon\in (0,1)$ there exists a $poly(\log n, \frac{1}{\varepsilon})$-round $(1-\varepsilon)$-approximation algorithm for MWM in {\bf CONGEST} provided that there is a similar algorithm for the largest cardinality matching problem.
\end{theorem}

\begin{proof} Consider the Algorithm \ref{algMain}.

 \begin{algorithm}[htbp]
  
  \caption{The main algorithm {\bf MaxWeightedMatchingMain}$(G, w, \varepsilon )$.}\label{algMain}
\begin{algorithmic}[1]
\STATE {\bf Input:} A graph $G$ on $n$ vertices, an edge-weight function $w:E(G)\rightarrow Q^+$, and a rational $\varepsilon\in (0,1)$.
\STATE {\bf Output:} A matching $M$ of the graph $G$.

\medskip

\medskip

\medskip

\STATE Define: $\varepsilon_1:=\frac{\varepsilon}{2}$, $\tau:=\frac{1}{1-\varepsilon_1}=\frac{1}{1-\frac{\varepsilon}{2}}$.

\STATE Each vertex $v$ rounds the weight $w(e)$ of an edge $e$ incident to it, to the closest power of $\tau$. That is, if $\tau^{j-1}<w(e)\leq \tau^{j}$, then $v$ defines the new edge-weight of $e$ as $w'(e)=\tau^j$. 


\STATE return {\bf MaxWeightedMatching}$(G, w',  \varepsilon_1, r )$.

\end{algorithmic}
\end{algorithm}

Now we are going to present the analysis of the algorithm. First of all, let us start with its correctness. Note that operations described in STEP 3, and STEP 4 can be done easily in {\bf CONGEST} since each vertex can round the edge-wight of an edge incident to it in one round. Next, note that though not every vertex of $G$ knows the $w$-weight of every edge, after this step, every vertex knows $\tau$ and $r$, hence it knows all possible edge-weights in $w'$. Since $W$ is known to all vertices, $r$, an upper bound for the number of different edge-weights in $w'$ can be computed by all vertices. Thus, the correctness of our main algorithm directly follows from Lemma \ref{lemmaRoundingUp} and Lemma \ref{lemmatdifferentFIXEDEdgeWeightsCongest}. 

Now, let us turn to upper bounding the number of rounds. Note that the number of rounds is at most \[\leq 1+poly(\log n, \frac{1}{\varepsilon}, r, \ln w_r),\]
where $r$ is the number of different edge-weights in the new instance, and $w_r$ is the largest edge-weight in $w'$. We have
\[\tau^{r-1}<W\leq \tau^r=w_r.\]
Hence
\[r< 1+\log_{\tau} W = 1+\frac{\ln W}{\ln \tau}\leq 1+{\left(\ln B+C\ln n\right)}\cdot \left(A+\frac{2}{\varepsilon}\right)=poly'\left(\log n, \frac{1}{\varepsilon}\right)\]
and
\[\ln w_r=r\cdot \ln \tau=r\cdot \ln \left(\frac{1}{1-\frac{\varepsilon}{2}}\right)\leq poly''\left(\log n, \frac{1}{\varepsilon}\right)\]
by Observation \ref{observationLogFunction}. Hence, we immediately get that the number of rounds is upper bounded by $poly'''\left(\log n, \frac{1}{\varepsilon}\right)$. The proof of the theorem is complete.
\end{proof}

In \cite{FMU21} the following result is announced:

\begin{theorem}
\label{theoremFMU}(\cite{FMU21}) For all $\varepsilon \in (0,1)$ there exists a $poly(\log n, \frac{1}{\varepsilon})$-round $(1-\varepsilon)$-approximation algorithm for the Largest Cardinality Matching problem in {\bf CONGEST}.
\end{theorem} Combined with Theorem \ref{theoremMain}, one can deduce that there is a $poly'(\log n, \frac{1}{\varepsilon})$-round $(1-\varepsilon)$-approximation algorithm for the Maximum Weighted Matching problem in {\bf CONGEST}.

\section{Implementing our algorithm in constant rounds in {\bf CONGEST}}

In this section, we show that we can reduce MWM to the largest cardinality matching in constant (actually, four) rounds in {\bf CONGEST}. 

The main idea is the following. In Round 1, every vertex $u$ of $G$ sends its label/ID to the neighboring vertices. This allows every vertex to learn the labels of all vertices adjacent to it. This requires just one round, and after that the vertices agree that the weight of any edge $e=uv$ will change the vertex with the smallest label.

In Round 2, every vertex can do the local rounding required in STEP 4 of {\bf MaxWeightedMatchingMain}. Note that after this step, the parameters $(\tau_1=\tau, t_1=t, \varepsilon_1=\frac{\varepsilon}{2}\in (0,\frac{1}{2}))$ can be assumed known to all vertices of the graph. This is because every vertex knows $\varepsilon$ and $W$-the largest edge-weight in $G$.

In Round 3, every vertex can do the local rounding required in STEP 13 of {\bf MaximumWeightedMatching}. Note that after this step, the parameters $(\tau_2=x, t_2=t', \varepsilon_2=\varepsilon'\in (0,1))$ can be assumed known to all vertices of the graph. After this round, note that {\bf MaximumWeightedMatching} becomes just a simple FOR loop, in which every vertex {locally and independently} computes the final approximation parameter in the instance of Largest Cardinality Matching. Note that the vertices just need to compute $\varepsilon_{final}$, and no exchange of information is required except those of labels in Round 1.

{\bf Remark:} Note that since the vertices do not exchange information after Round 1, the fact whether there is an edge $e$ with $w(e)>W$ or not, are locally the same (from the perspective of the vertex). Hence, if we have an upper bound for $t_1$ (or $t_2$), without loss of generality, we can assume that $t_1$ (or $t_2$) coincide with the bound. And each vertex continues its FOR loop till this parameter even if initially the vertex was not adjacent to an edge with this weight.

Finally, we would like to discuss the following issue on presenting rational numbers in our algorithm. Since we assumed that rational numbers are presented as a pair of integers that are bounded by $poly(n)$, we have to make sure that in our algorithm we do not have a rational number violating this constraint. Note that in our description of the algorithm in constant rounds, no rational number is exchanged between vertices. Hence, all these numbers are computed by a vertex and they are not involved in other {\bf CONGEST} algorithms except $\varepsilon_{unweighted}$ that is part of the algorithm from \cite{FMU21}. However, we can easily overcome this problem by simply assuming that the actual value of $\varepsilon_{unweighted}$ is $\frac{1}{2^k}$, where $k\in Z^+$ and \[\frac{1}{2^k}\leq \varepsilon_{unweighted} < \frac{1}{2^{k-1}}.\]
For the definition of $\varepsilon_{unweighted}$ see the end of the proof of Lemma \ref{lemmatdifferentFIXEDEdgeWeightsCongest}.

\section*{Acknowledgement}
\label{AcknowledgeSection} The author would like to thank Boston College, its Computer Science department and Dr. Hsin-Hao Su for hospitality.

%



\bibliographystyle{elsarticle-num}



\end{document}